\documentclass[10pt, conference,onecolumn]{IEEEtran}
\IEEEoverridecommandlockouts
\usepackage{setspace}
\usepackage{cite}
\usepackage{amsmath,amssymb,amsfonts}
\usepackage{graphicx}
\usepackage{caption,subcaption}
\usepackage{dblfloatfix}
\usepackage[table,xcdraw]{xcolor}
\usepackage{booktabs}
\usepackage{textcomp}
\usepackage{soul}
\usepackage[long]{optidef}
\usepackage{todonotes}
\usepackage{ifthen}
\usepackage{listings}
\usepackage{diagbox}
\usepackage{slashbox}
\usepackage{adjustbox}
\usepackage{comment}
\usepackage{amsthm}

\newtheorem{lemma}{Lemma}

\usepackage{balance}

\usepackage{varwidth}
\usepackage{algorithmic}

\definecolor{yellow}{rgb}{1.0, 0.98, 0.596}
\newboolean{ncsversion}
\setboolean{ncsversion}{true}

\usepackage{orcidlink}
\hypersetup{%
    pdfborder = {0 0 0}
}
\usepackage{xcolor}
\usepackage[ruled,vlined]{algorithm2e}

\def\BibTeX{{\rm B\kern-.05em{\sc i\kern-.025em b}\kern-.08em
    T\kern-.1667em\lower.7ex\hbox{E}\kern-.125emX}}
\begin{document}



\title{Accelerating End-host Congestion Response using P4 Programmable Switches}

 \author{\IEEEauthorblockN{Nehal Baganal-Krishna\,\orcidlink{0000-0002-9099-1754}\IEEEauthorrefmark{1}\IEEEauthorrefmark{5}, Tuan-Dat Tran\IEEEauthorrefmark{5}, Ralf Kundel\IEEEauthorrefmark{3}, Amr Rizk\,\orcidlink{0000-0002-9385-7729}\IEEEauthorrefmark{5}}
  \IEEEauthorblockA{\IEEEauthorrefmark{1}University of Ulm, Germany}
 \IEEEauthorblockA{\IEEEauthorrefmark{5}University of Duisburg-Essen, Germany}

 \IEEEauthorblockA{\IEEEauthorrefmark{3} Techincal University of Darmstadt, Germany}
 }

\maketitle

\begin{abstract}
\noindent
\textbf{Transport layer congestion control relies on feedback signals that travel from the congested link to the receiver and back to the sender.
This \textit{forward congestion control loop}, first, requires at least one rount-trip time (RTT) to react to congestion and secondly, it depends on the downstream path after the bottleneck.
The former property leads to a reaction time in the order of RTT + bottleneck queue delay, while the second may amplify the unfairness due to heterogeneous RTT.}
\textbf{
In this paper, we present Reverse Path Congestion Marking (RPM) to accelerate the reaction to network congestion events without changing the end-host stack.
RPM decouples the congestion signal from the downstream path after the bottleneck while maintaining the stability of the congestion control loop.
We show that RPM improves throughput fairness for RTT-heterogeneous TCP flows as well as the flow completion time, especially for small Data Center TCP (DCTCP) flows.
Finally, we show RPM evaluation results in a testbed built around P4 programmable ASIC switches.}


\end{abstract}

\begin{IEEEkeywords}
\noindent Congestion Marking, AQM, P4, Data plane programming
\end{IEEEkeywords}

\section{Introduction}

Bufferbloat describes the excessive delays packets observe when traversing links with deep queues~\cite{gettys2011bufferbloat}.
Large buffer memory in combination with the available bandwidth and congestion control mechanisms of most of the modern (transport) protocols such as TCP (e.g. Cubic, BBR) and QUIC lead to large packet delays.
This is due to the available bandwidth estimation and congestion control mechanisms actively inducing congestion signals at the bottleneck to estimate the appropriate sending rate~\cite{huang2012confused}.
Active Queue Management (AQM) tackles this problem by keeping the buffer filling small.
This is done by sending congestion signals, i.e., through packet drops or ECN markings, when the delay due to buffering exceeds a target value.
Prominent AQM algorithms are Codel\cite{nichols2012CoDel}, PIE\cite{pan2013pie} and RED\cite{floyd1993red} and their variations.

Although AQM aims to manage the Bufferbloat problem, there still remain two fundamental problems that are due to (i) the end-to-end construction of the current congestion control mechanisms and (ii) traffic generation patterns of modern applications, especially in data centers.
Prevalent congestion control mechanisms are self-clocked, i.e., they work on the time scale of roundtrip time (RTT) with congestion signals created at bottleneck links having to bounce back from the receiver to the sender.
This binds the sender's reaction time to congestion to this time scale which is often used as a cornerstone to showing the stability of congestion control~\cite{srikant2014communication}.
However, it is also known that this RTT time scale binding of congestion control leads to throughput unfairness when multiple flows of different RTTs compete for the bandwidth at the same bottleneck~\cite{floyd1991connections,barakat2000tcp,altman2000fairness}.
Note that the end-to-end argument in congestion control also dictates that the excessive queueing time at the congested link buffer becomes part of the time required for the congestion signal to travel back to the sender.

The second problem arises especially in data centers where most traffic flows are short, e.g., due to application request-response patterns~\cite{alizadeh2010data, cho2017credit}.
These short flows, combined with the very high link capacities found in data center fabrics~\cite{uyeda2011efficiently}, lead to so-called microbursts.
These $\mu$bursts last shorter than a single RTT \cite{alizadeh2012less}, but most importantly, microbursts can lead to excessive packet drops as data center buffers are typically very small~\cite{ren2014survey}.
It is evident that classical congestion control mechanisms struggle with such phenomena leading to bloated flow completion times~\cite{zats2012detail}.
However, approaches that change the congestion control stack tackle this problem at the cost of special purpose end-host stacks~\cite{zhang2005xcp}.

In this paper, we propose an approach to reduce the congestion reaction time as well as the flow completion time for short flows using in-network support~\cite{Kundel2021ECN}.
We denote this approach Reverse Path Congestion Marking (RPM)\footnote{RPM stands on the shoulders of approaches for out of band congestion notification using special packets, e.g. ICMP or others \cite{feldmann2019, salim98}.}.
RPM modifies the in-flight ACKs for fast and stable congestion control in sub-RTT time scales.
RPM is compatible with TCP and DCTCP, hence no modifications to end-hosts are required.
Finally, we implement and show RPM in a programmable data-plane using P4 on modern networking hardware (Intel Tofino switch).

To summarize, our contributions are as follows:
\begin{itemize}
    \item We illustrate an approach for in-band reverse path congestion marking without additional control traffic.
    \item We provide an analytical model of RPM, proving its stability.
    \item We show an implementation of RPM in P4 on programmable networking hardware, including results from a testbed deployment.
\end{itemize}

\section{Challenges for forward-path congestion control loops}
\label{sec:challenges}
Congestion control closed-loop mechanisms comprise two major steps, i.e. (i) the congestion point signals congestion to the receiver using packet drops or CE code points in the IP header, and (ii) the sender reduces the congestion window conforming to the congestion notification received.
The above-mentioned control loop has two major challenges, i.e., congestion signal delay and path dependency, which we discuss in the following.

 \subsection{{Delayed congestion control response}}
The congestion control loop requires at least one RTT or even more time to ease the congestion at the bottleneck after congestion detection. To reach the sender, the congestion signal has to traverse from the bottleneck to the receiver and then from the receiver back to the sender. The propagation delay of the forward loop from the bottleneck to the sender is  $\tau_2 + \tau_3 + \tau_4$, as depicted in Fig.~\ref{fwd_ecn}. The congestion control loop additionally requires the time $\tau_1$ to propagate the effect of a reduced congestion window after the congestion signal arrives at the sender if this reduction is executed immediately. Hence, the total congestion signal propagation time elapsed from congestion occurrence in the bottleneck until the bottleneck experiences a reduction in the sending rate is at least $\sum_{i=1}^4 \tau_i$ (cf. Fig.~\ref{fwd_ecn}). The delay of the congestion control loop leads to a delayed dynamical system at the output port queue experiencing congestion.

 \subsection{{Downstream path dependent congestion response}}
The second challenge in classical forward-path  congestion control loops is the dependency of the reaction time of the sender node on the downstream path from the bottleneck to the receiver.
Considering a network where multiple flows with different RTTs share one bottleneck, current congestion control approaches lead to unfair bottleneck capacity distribution among these flows~\cite{rfc2582,brown2000resource,ha2008cubic,lakshman1997performance}.
Hence, although the flows may experience congestion simultaneously, the time of arrival of the congestion notification at the respective sources varies and hence they react with different delays to the congestion leading to throughput unfairness with sources having a smaller RTT receiving more bandwidth.

From the two challenges above, we observe that end-to-end congestion control without network support is difficult to configure.
TCP friendliness~\cite{floyd1991connections,barakat2000tcp,altman2000fairness} and max-min fairness~\cite{srikant2014communication} are much harder to obtain if the congestion control loops have heterogeneous delays and do not only depend on the shared bottleneck but also on their own paths.
To this end, we present a design of a data-plane based support for congestion marking that aims to dispense (in parts) with the two challenges above in the following sections.

\begin{figure}[t]
\centering
\includegraphics[width=0.8\columnwidth]{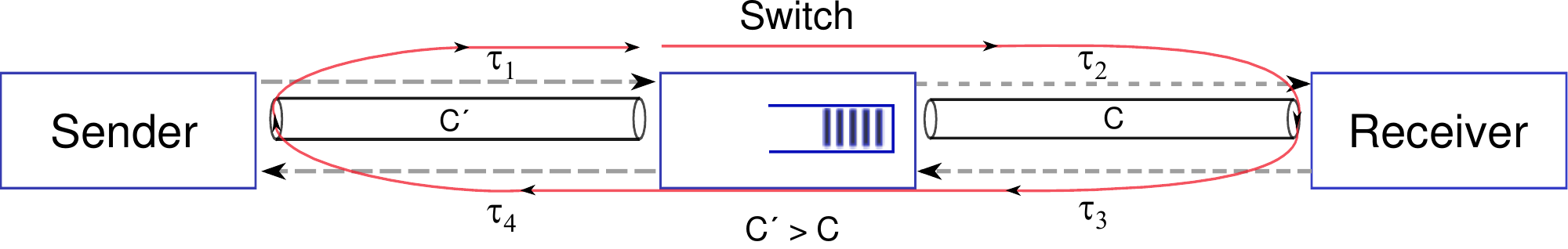}
\caption{Forward path ECN feedback loop.}
\label{fwd_ecn}
\end{figure}

\section{Design of Reverse Path Congestion Marking}
In this section, we first present the design objectives behind Reverse Path Congestion Marking (RPM), then present its design components before discussing these.

\subsection{RPM Design objectives}

To overcome the challenges discussed in the Sect.~\ref{sec:challenges}, we first derive the objectives for RPM that act as a framework for the subsequent implementation.
\subsubsection{Decouple congestion signals from the downstream path} In the current state-of-the-art congestion control loop, for the congestion notification (either duplicate ACKs or ECE marked ACKs) to reach the sender, the signal must traverse the path beyond the bottleneck to the receiver and then take the path back to the sender. If the downstream network after the bottleneck is dynamic, the time taken by the signal to reach the receiver generally increases (on average) and its variability increases too.

Any additional time elapsed after the creation of a congestion signal leads to additional buffer filling at the bottleneck and causes an increase in flow packet drops as well as flow completion time.
Therefore, it is necessary to decouple congestion signals from the downstream path and its bottlenecks, further disassociating the dependence of the sender's reaction to the congestion signal on the downstream path. 

\subsubsection{Scalable data plane implementation}  Common models for the analysis of AQM mechanisms, such as CoDel, consider a segregation of flows that traverse the same bottleneck to correctly inform the sender of congestion.
As we consider scenarios where possibly hundreds or thousands of flows traverse a  bottleneck, and the storage of a per-flow AQM state is impractical on modern data plane implementations, it is necessary to obtain an RPM algorithm implementation of the data-plane that is flow-agnostic.

\subsubsection{No modifications to the end-host stack}  Transmission Control Protocol (TCP) and the Data Center TCP (DCTCP) are two of the most widely used transport layer protocols.
Modifications to the end-host TCP or DCTCP congestion control stack are weary to standardize.
To allow seamless adoption, the RPM algorithm implemented in the network data plane should adhere to TCP and DCTCP specifications without requiring changes to the end-hosts.

\subsection{RPM Design Components}

\subsubsection{Reverse path ECN-based Marking}

The idea of reverse path congestion marking is to mimic the ECE flag marking functionality of the receiver in the data plane.
Replicating the receiver's marking scheme onto the switch improves the congestion indication flexibility, as it can notify congestion by directly marking the ECE field of the in-flight ACKs going back to the sender.
The reverse marking feature shortens the reaction time of the congestion control loop by decoupling the congestion signal from the downstream bottleneck.
Figure~\ref{concept} illustrates the approach to Reverse Path Congestion Marking on the data plane.
Congestion detection and marking are carried out by an AQM such as CoDel, e.g., using the data plane implementation~\cite{kundel18}.
The reverse path marking on the data plane requires a congestion state and a reverse marking module, which is explained below in more detail.

In any ECN-enabled TCP congestion control loop, when the receiver accepts a packet with marked CE code points, it sets ECE flags in all subsequent ACKs to the sender to report congestion~\cite{rfc_ecn}.
The receiver continues to mark ACKs until it receives a packet with a congestion window reduced (CWR)
flag set, indicating a reduction in the congestion window at the sender. Even though the sender receives multiple ACKs with ECE set, it halves its congestion window only once per RTT.
DCTCP clients apply similar functionality with an additional single-bit state variable denoted DCTCP Congestion Encountered (DCTCP.CE) to enable delayed ACKs~\cite{rfc_dctcp}.
When a DCTCP receiver receives a packet with the CE code points set, DCTCP.CE is set to "1" and transmits an ACK with the ECE flag set to the sender.
Here too, the receiver continues to mark ACKs until it receives a packet with a CWR
flag set thereby setting DCTCP.CE to ``0''. The DCTCP sender also reduces the congestion window once per RTT but the reduction depends on the number of ACKs (with ECE flag) the sender receives.

When the AQM on the switch data plane (see Fig.~\ref{concept}) detects congestion, the marking mechanism of RPM sets the ECE bit in the in-flight ACKs of that flow, to notify the sender.
To identify this flow and the corresponding ACKs, a reverse marking module is required to keep a minimum state mapping flows to a binary congestion observed value.
This can be realized on modern P4 programmable switches (such as the Intel Tofino ASIC) and will be discussed in the next section.
For every congestion detected by the AQM for a specific flow at the bottleneck, one in-flight ACK of that flow on the reverse path will be marked (by setting the ECE flag).
The design of the marking operation at the switch is the same for TCP and DCTCP, although the DCTCP receiver employs specific state variables.
Note that the data plane need not keep a state of DCTCP Congestion Encountered, which renders RPM to be TCP and DCTCP compatible simultaneously.

\subsubsection{Per-Flow State Aggregation}

     When the flows pass through the switch and experience congestion, the switch should notify congestion to the senders of each flow.
     Hence, the switch has to remember the flow experiencing congestion and mark the ACK corresponding to the stored flow. The Congestion State Register in the data plane holds a memory unit for each flow traversing through the switch. When an AQM detects congestion for a flow, the congestion state register is incremented at a specific location pointed by the flow.
     Subsequently, when the corresponding in-flight ACK of the previously congested flow enters the switch, the congestion state register is decremented and ECE flag is marked.
     The per-flow congestion state facilitates RPM to uniquely notify the sender of a flow experiencing congestion and helps RPM achieve simultaneous compatibility with both TCP and DCTCP.

\subsection{Discussion}

RPM isolates the dependency of congestion notification on the downstream path after the bottleneck, as the marking algorithm on the data plane marks the in-flight ACKs going towards the sender instead of packets in the forward direction, which may experience different and variable delays towards the receiver and back.
Hence, the time required for the congestion feedback to reach the sender from the bottleneck is minimized.

One prerequisite of RPM is that the route taken by the data flow to reach the receiver and the route taken by the ACK to reach the sender should be the same, or at least the bottleneck switch should be on both paths forward and reverse.
Checking this condition is straightforward with an RPM data plane implementation using P4 such that RPM can be turned off for flows that do not adhere to this condition.

Also note that the size of the register present in the P4 programmable switch, e.g., Intel Tofino, dictates the number of flows for which RPM is operational. We report on the utilized on-switch memory in Sect.~\ref{sec:implementation_eval}.

The scarce memory availability on the data plane adds an upper threshold to the number of flows passing through the single switch that RPM can handle at any given time, as each flow requires one bit in the congestion state register.
Instead of storing information for every flow that passes through the bottleneck, we can choose to store per port congestion information.
All flows routing through a single bottleneck switch port will be considered one port group of flows for which RPM holds one state.
We note that this marking scheme might be too aggressive towards congestion window reduction.

RPM works with TLS but is difficult to implement with IPSec. In TLS, as the TCP payload is encapsulated and RPM can mark the ECE field of the TCP header in case of congestion; however, with IPSec the TCP header is encapsulated, which means RPM cannot mark ACKs if RPM is not running at the IPSec tunnel ingress.




\begin{figure}[t]
\centering
\includegraphics[width=1\columnwidth]{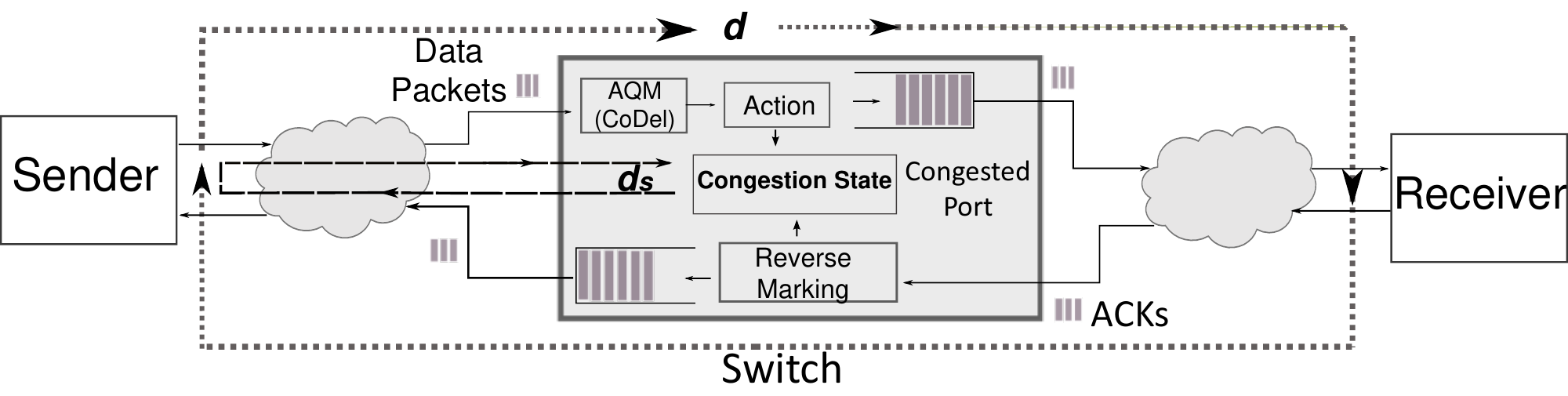}
\caption{Components of Reverse Path Congestion Marking.}
\label{concept}
\end{figure}

\section{Stability of Reverse Path Congestion Marking }

In this section, we provide an analytical model of  Reverse Path Congestion Marking based on the sketch of the setup depicted in Fig.~\ref{model}.
We assume that the TCP source on the sender (left-hand side) is in the congestion avoidance state of an Additive increase and Multiplicative decrease (AIMD) congestion control algorithm.
We assume that the bottleneck switch signals congestion by marking CE code points in the IP header and that the bottleneck link has a capacity of $c$ bps.
The following model extends the forward path congestion control analysis in \cite{kumar2004communication}.
This extension is not trivial as the resulting sending window ODE has multiple, different delay factors making the analysis much harder than, e.g., in~\cite{kumar2004communication}.

\begin{figure}[t]
\centering
\includegraphics[width=0.8\columnwidth]{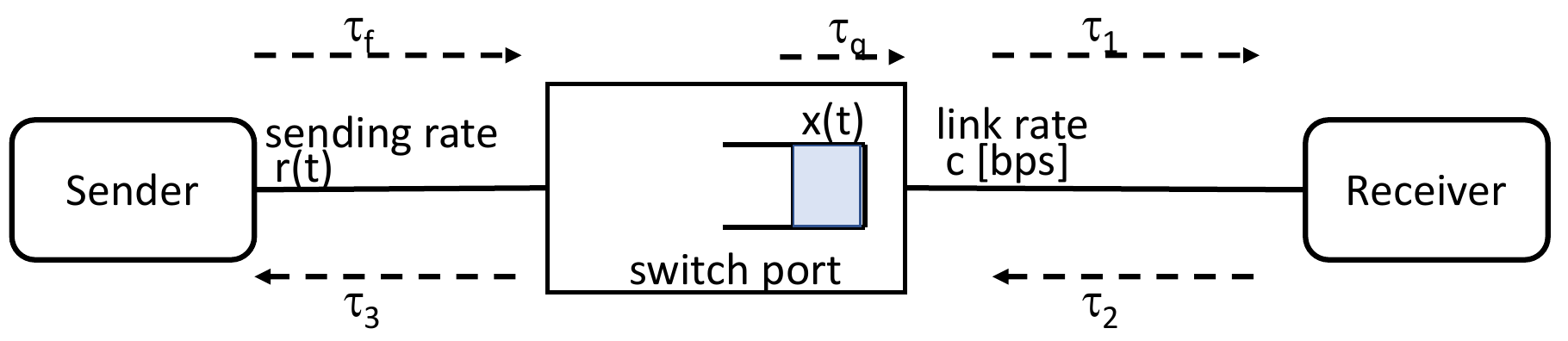}
\caption{Deterministic dynamic system model of the congestion control loop of a single connection over a bottleneck link.}
\label{model}
\end{figure}
We consider the scenario where an ACK packet traversing back to the sender without congestion notification (marked ECE bits in the TCP header) causes the time-dependent congestion window function $w(t)$ to increase by $\frac{a}{w(t)}$, where $a$ relates to the additive increase factor.
Now, we consider an AQM that randomly marks packets with congestion bits given that the queue at the switch output port experiences congestion, i.e., the queue length is larger than an AQM target value of $x^*$.
Precisely, the CE code points are marked $1$ with probability $\eta [x(t) - x^{*}]^{+}$ and are marked $0$ with $(1-\eta) [x(t) - x^{*}]^{+}$\footnote{We use $[x]^{+}$ to denote $\max\{x,0\}.$}.
Note that the congestion marking probability depends on the queue length exceeding the target $ [x(t) - x^{*}]^{+}$ and the normalization factor $\eta$ controls the ``aggressiveness'' (or slope) of the marking with the grade congestion.
As depicted in Fig.~\ref{model} we call the forward propagation delay to the switch $\tau_{f}$ and propagation delay for the rest of the loop $\tau_{r} = \sum_{i=1}^{3}\tau_i$. Note that the queueing delay at the switch port is denoted $\tau_q$ such that we define $\tau_f + \tau_r := \tau$ and the round trip delay $d:= \tau+\tau_q$.

When RPM indicates congestion, the switch sets the ECE bit in the ACK of the corresponding flow, which experiences a bottleneck.
Hence, the congestion signal (ECE bit) experiences a propagation delay back of approximately $ \tau_{rs} = \tau_3$ and the short control loop round trip time will be $ \tau_s = \tau_f + \tau_{rs}$. We note that at the time of congestion experience there might not be an ACK travelling back, however, (i) $\tau_{rs} < \tau_{r}$ always holds, and (ii) we approximate the time until an ACK is seen on the reverse path to zero. Once the sender receives an ACK with an ECE bit set, the sender's congestion window decreases multiplicatively by a factor of $b$.
Hence, we obtain the following rate of change of the congestion window $w(t)$ of the TCP sender
\begin{align}
\label{window_func}
    \dot w(t) =& \frac{a}{w(t)} r(t-\tau -\tau_{q})  (1-\eta  \left[x(t-\tau_{r}-\tau_{q})-x^{*}\right]^+ \nonumber\\
            & - b\cdot w(t)  r(t-\tau_{s})  \eta  \left[x(t-\tau_{rs})-x^{*}\right]^+
\end{align}
The first term on the right side is an additive increase rate each time the sender receives an unmarked ACK.
The rate of the arrival of unmarked ACKs at the source is $ r(t-\tau -\tau_{q})$, and a factor of $(1-\eta  (x(t-\tau_{r}-\tau_{q})-x^{*})$ ACKs are not marked by switch.
The second term represents the multiplicative decrease proportional to $b>0$ and the current window rate $w(t)$.
The arrival rate of congestion notifications is $  r(t-\tau_{s})$ as the propagation delay of ECE bit after congestion indication at the switch is $\tau_{rs}$ and a factor of $ \eta  (x(t-\tau_{rs})-x^{*}) $ ACKs are marked by switch.
Next, we use that $d = \tau + \tau_q > 0$ is the round trip delay for the connection and denote for nomenclature homogeneity $d_s = \tau_s$ as  the ``short'' round trip delay between the sender and the bottleneck.

\begin{lemma}[Stability of RPM]
\label{lem:stability_RPM}
Given a fluid AIMD TCP source with additive increase parameter $a$ (e.g. $1$) and multiplicative decrease parameter $b$ (e.g. $\frac{1}{2}$) using a network path utilizing Reverse Path Congestion Marking (RPM).
Assume the path has a forward path round-trip time $d$, a (reverse) short round trip time between the sender and the bottleneck $d_s$, and a bottleneck link capacity $c$ and define $ \gamma := \frac{abc}{a+bc^{2}d^{2}}$. The dynamical system describing the congestion window of the AIMD source in~\eqref{window_func} is stable with a positive normalization factor $\eta$ given by
\begin{equation}
\normalsize
\label{eta}
    \eta = \frac{(e^{-sd}-e^{-sd_{s}} - 2)s \gamma - s^{2}}{\frac{e^{-sd}}{d^{2}} + \frac{c^{2} e^{-sd_{s}}}{2}} \,,
\end{equation}
for any $s<s^*<0$ with $s^*$ given in \eqref{s_star}.
\end{lemma}%
\noindent The proof of Lem.~\ref{lem:stability_RPM} is given in the appendix.

Based on Lem.~\ref{lem:stability_RPM} we can obtain the marking aggressiveness~$\eta$ of the RPM. As an exercise substituting common values for the system parameters one may find that $\eta$ is typically small and positive.

By varying $d_s$ (the position of bottleneck) with a constant $d$ and $c$ in (\ref{eta}), we find that $\eta$ is approximately constant. This indicates that the marking factor $\eta$ need not be changed based on bottleneck location in the network to maintain stability.

\section{Implementation \& Evaluation}
\label{sec:implementation_eval}
In this section, we discuss the implementation details of RPM. We further indulge in discussing the testbed and evaluating RPM performance.
\subsection{Algorithmic implementation of RPM}

As indicated in Fig.~\ref{concept}, the implementation of RPM has three components, (1) congestion indication, (2) reverse path marking, and (3) congestion state clearance.
For congestion indication, the switch can use an arbitrary AQM; in this paper, we employ CoDel in the data plane~\cite{kundel18}. We modify this data-plane implementation of CoDel to support RPM.

When CoDel detects a bottleneck, RPM increments a Congestion State Register $\mathbf{R_{C}}$ at the location pointed by the flow's 5-tuple Hash function $h($IP$_{src}$, IP$_{dst}$, IP$_{prot}$, p$_{src}$, p$_{dst})$. The hash function $h$ uniquely identifies each flow and maps the congestion state register to a unique per-flow memory unit.
To mark in-flight ACKs, RPM reads the congestion state register at the location pointed by the flow's 5-tuple Hash function $h$, but the source-destination IP addresses IP$_{(\cdot)}$ and L4-ports p$_{(\cdot)}$ are swapped.
If the register $\mathbf{R_{C}}$ is greater than or equal to "1" indicating congestion, the ECE flag of the TCP header in the ACK packet is set and the register is decremented by "1".


The packet entering the data plane goes through all the RPM components mentioned above to seamlessly support TCP or DCTCP for the following reasons; first, the data plane need not differentiate between the data packets traveling from sender to receiver and ACKs traveling from receiver to sender.
In addition, TCP and DCTCP are bidirectional flows meaning both sender and receiver can send data to each other while simultaneously acknowledging the received data. So, a packet considered an ACK might carry data.
Hence, flows in the reverse direction, i.e., of the ACKs, can be marked.
Note that a TCP flow requires at least one ACK with ECE flag set to halve its congestion window per RTT, while a DCTCP flow exactly needs the count of ACKs the sender obtains to calculate congestion window per RTT~\cite{rfc_dctcp, alizadeh2010data}.
The marking scheme employed by RPM satisfies the congestion window calculation of both TCP and DCTCP.
This is because every time a flow experiences congestion, in-flight ACKs in the reverse direction of the flow will be marked.
Algorithm~\ref{algo_RPM} describes the core logic behind RPM when implemented in the data plane.


\begin{algorithm}[H]
{
\caption{Pseudocode of the RPM Algorithm}
\label{algo_RPM}
\begin{algorithmic}[1]
\STATE Run CoDel

\tcp{CoDel congestion state denoted $\text{drop}_C$}
\IF {$\text{drop}_C == \text{True}$}
    \STATE $pos \gets h($IP$_{src}$, IP$_{dst}$, IP$_{prot}$, p$_{src}$, p$_{dst})$
    \STATE $\mathbf{R_{C}}[pos] \gets \mathbf{R_{C}}[pos] + 1$
\ENDIF

\tcp{When an ACK arrives}
\STATE $pos \gets h($IP$_{dst}$, IP$_{src}$, IP$_{prot}$, p$_{dst}$, p$_{src})$
\IF{$\mathbf{R_{C}}[pos] \geq 1$}
    \STATE $\text{TCP.ECE} \gets 1$
    \STATE $\mathbf{R_{C}}[pos] \gets \mathbf{R_{C}}[pos] - 1$
\ENDIF
\end{algorithmic}
}
\end{algorithm}

\begin{figure}[h]
\centering
\includegraphics[width=0.7\columnwidth]{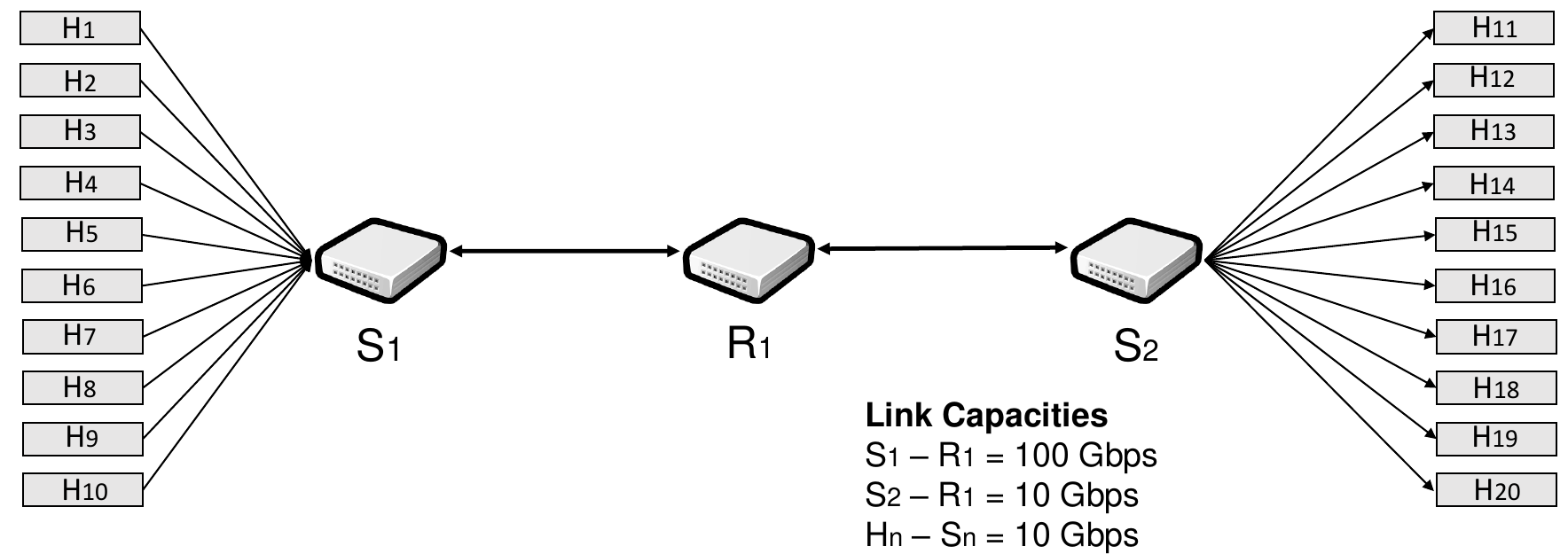}
\caption{Testbed Topology.}
\label{fig:ISP_testbed}
\end{figure}

\subsection{Evaluation Results}

Next, we compare RPM with CoDel-Forward Congestion Marking (CoDel-FWD).
We use the topology in Fig.~\ref{fig:ISP_testbed} to study the performance of RPM.
First, we study the throughput fairness when the flows in the network have different RTTs but share the same bottleneck. 
The hosts $H_1$ to $H_{10}$ act as senders and hosts $H_{11}$ to $H_{20}$ act as receivers, such that $H_i$ sends its traffic only to $H_{i+10}$. Host links are set to $10\ Gbps$.
The hosts may be run using high performance traffic generation as in~\cite{KundelSHRK22}.
In CoDel-FWD, after congestion identification by the CoDel algorithm, the marking scheme directly sets the CE code points in IP header, indicating congestion to the receiver.
The CoDel parameters, i.e., the target delay and the CoDel interval, are set to $1\ ms$ and $20\ ms$. As the Tofino switch has finite memory, the buffer has an upper limit on the number of packets it can store, which restricts the delay in this configuration to be below $2\ ms$.

\begin{figure}[t]
\centering
\includegraphics[width=0.7\columnwidth]{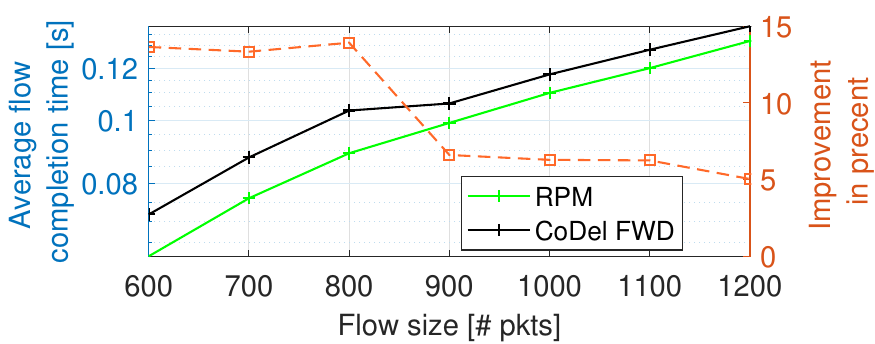}
\caption{Average flow completion time for short flows. The shorter the flows the more RPM helps decreasing the flow completion time due to shorter reaction time compared to CoDel on the forward path. }
\label{fig:FCT}
\end{figure}

\paragraph*{Fairness Experiment} For the fairness experiments we set the delay of each sender-receiver pair link using the Traffic Control (tc) linux command on both senders and receivers.
The link delays of the links $H_i-S_1$ are all set to $10\ ms$ and the link delays of the links  $S_2-H_j$ for each experiment conducted are given in Tab.~\ref{tab:experiments}.
The P4 switch $R_1$ (in which AQM is deployed) is connected to $S_1$ with a $100\ Gbps$ link and to $S_2$ with a $10\ Gbps$ link leading to congestion at the output port of $R_1$. We deploy AQM on $R_1$. We measure Jain's fairness index\footnote{Jain's fairness index defined as $J = \frac{E[X]^2}{E[X^2]}$ \cite{jain1999throughput}.} of the flow throughput for long TCP Cubic flows.

\begin{table}[h]
\centering
{

     \caption{Experiment configurations of receiver link delays.}
     \label{tab:experiments}
 \scalebox{1}{
 \begin{tabular}{|c||*{5}{c|}}\hline
 \makebox[4em]{\textbf{Experiment}}&\makebox[4em]{\textbf{$S_2$-$H_{11/12}$}}&\makebox[4em]{\textbf{$S_2$-$H_{13/14}$}}&\makebox[4em]{\textbf{$S_2$-$H_{15/16}$}}&\makebox[4em]{\textbf{$S_2$-$H_{17/18}$}}&\makebox[4em]{\textbf{$S_2$-$H_{19/20}$}}\\\hline\hline
  $1$& $10\ ms$ & $10\ ms$ & $10\ ms$ & $10\ ms$ & $10\ ms$ \\\hline
 $2$ & $10\ ms$ & $20\ ms$ & $30\ ms$ & $40\ ms$ & $50\ ms$\\\hline
  $3$& $20\ ms$ & $40\ ms$ & $60\ ms$ & $80\ ms$ & $100\ ms$\\\hline
 \end{tabular}}
 }
\end{table}

\begin{table}[h]
\centering
\caption{Jain's Fairness index comparison. $J\rightarrow 1 $ indicates complete fairness. RPM increases throughput fairness for heterogeneous flow RTTs (Experiments 2/3).}
     \label{jain-fair}
 \scalebox{1}{
 \begin{tabular}{|c||*{4}{c|}}\hline
 \makebox[6em]{\textbf{Experiment}}&\makebox[6em]{\textbf{J(CoDel-FWD)}}&\makebox[6em]{\textbf{J(RPM)}}\\\hline\hline
  $1$& $0.89$& $0.89$\\\hline
 $2$&$0.74$&$0.81$\\\hline
  $3$&$0.82$&$0.86$\\\hline
 \end{tabular}}

\end{table}

Table~\ref{jain-fair} shows the results of our fairness experiments. Experiment 1, that is used for calibration, sets homogeneous link delays of 10ms on all links. We observe that both RPM and CoDel-FWD obtain the same fairness index of the flow throughput. Experiments 2 and 3 show that when the RTTs are heterogeneous, RPM improves the throughput fairness as all the flows have the same RPM reaction time of $20$ms. Note that RPM cannot improve the RTT unfairness due to differing RTTs on the first link, i.e. $H_i - S_1$.
Also, RTT unfairness in the additive part of AIMD remains unchanged.


\paragraph*{Flow Completion Time Experiment}
Here, we evaluate the average flow completion time, especially of short flows, as we expect RPM to perform best with short flows as it shortens the reaction time.
We deploy RPM/CoDel\_FWD on switch $S_1$. We set the link capacity of all links to $10\ Gbps$ and remove the link delays introduced in the fairness experiment.

Figure~\ref{fig:FCT} shows the average flow completion time for short flows given in the number of MSS packets. As RPM helps to decrease the reaction time of the congestion control loop, it is especially beneficial for short flows if they suffer from congestion losses.

\section{Related Work}

Established AQM schemes, such as Random Early Detection (RED) \cite{floyd1993red}, CoDel\cite{nichols2012CoDel}, and PIE\cite{pan2013pie}, send forward path congestion notifications such that the sender adjusts its sending rate. The congestion signal is either packet drop or setting ECN flelds in IP header.
when CE code points of IP header are set by the router in case of congestion \cite{RFC3168} the receiver reflects this information in the ECN-Echo bit of the TCP header of the ACK to notify the sender. Once the sender receives this ACK with ECN-Echo bit set, it reduces its sending rate. The minimum time required to observe a reduced sending rate by a congested port, i.e, the reaction time, is at least one RTT.

To decrease this reaction time, Backward ECN (BECN) for network congestion was originally proposed for ATM networks~\cite{salim98}, wherein the router informs the sender about the congestion by sending out-of-band ICMP Source Quench (ISQ) messages without the intervention of the receiver. This reduces the reaction time but generates additional traffic in the feedback loop and loss of the feedback ISQ packets would not be noted by the sender.

EXplicit Congestion Protocol (XCP) by Zhang and Henderson \cite{zhang2005xcp} proposed a constant exchange of information between router and end hosts. The routers notify the receiver about the state of congestion in the network by marking it in the XCP header. The sender receives congestion notifications from the receiver through ACKs. This process takes around one RTT and involves modifying the end-hosts.

Similar to BECN, the work in~\cite{feldmann2019} proposed a network-assisted congestion feedback using NACK packets. Instead of sending ISQ packets, this approach sends NACKs generated by programmable switches to notify the sender regarding congestion. Here too, the sender stack has to be modified to understand the NACKs and the additional out-of-band backward traffic remains.

One approach to use in-flight ACKs on their way to the sender is found in different forms in~\cite{raqm} and \cite{meng2022achieving}. In both works the ACKs are delayed in the router. This expires the TCP Retransmission Timeout at the sender, forcing the sender to reduce its sending rate. Both works~\cite{raqm} and \cite{meng2022achieving} estimate the latency between the access point and the receiver by taking the queue delay as well as transmission delay into account. The work in~\cite{meng2022achieving} is only confined to wireless communication with last-mile access points, whereas~\cite{raqm} mainly tackles incast problems in multi-tenant data centers. In contrast, RPM is a congestion marking scheme that can be combined with any ECN-enabled congestion control at the receiver and any AQM algorithm deployed on the switch to mark in-flight ACKs to improve the responsiveness of the congestion control loop.

\section{Conclusion}

In this paper, we presented Reverse Path Congestion Marking (RPM), i.e., an in-network method to improve the congestion control reaction time without changing the end-host stack.
With RPM, we can send congestion signals directly back to the sender, decoupling congestion signaling from the downstream path past the bottleneck.
We mathematically prove that RPM is stable.
Further, using a testbed built around P4 programmable ASIC switches, we showed that it improves throughput fairness for RTT-heterogeneous TCP flows as well as the flow completion time, especially for small DCTCP flows.


\section{Appendix}
\begin{proof}[Proof of Lemma~\ref{lem:stability_RPM}]
By Little's theorem we know that the average congestion window equals the average sending rate times the average delay.
With this in mind we approximate the sending window, hence, in the case of unmarked ACKs $w(t) = r(t)  d$ and in the case of marked ACKs we have $w(t) = r(t) d_{s}$.
we can rewrite \eqref{window_func} as
\begin{align}
\Large
\label{rate_func}
    \dot r(t) =& \left[\frac{a}{r(t)  d^{2}} r(t-d)\right](1-\eta \left[(x(t-d+\tau_{f})-x^{*}\right]^{+})\nonumber \\
            & -\left[ b \cdot r(t)r(t-d_{s})\right]\eta\left[x(t-d_{s}+\tau_{f})-x^{*}\right]^{+}
\end{align}
By linearizing \eqref{rate_func}, substituting the dynamics of $x(t)$, and using $r(t)$ again, we obtain the following delay differential equation
\begin{align}
\Large
\label{simpl_d_d_r(t)}
        \ddot r(t) &+ \gamma \left[2 \dot r(t) - \dot r(t-d) + \dot r(t - d_{s}\right] \nonumber\\
        &+ \alpha r(t-d) + \Omega r(t-d_{s}) = 0
\end{align}
where $ \gamma = \frac{abc}{a+bc^{2}d^{2}}$, $ \alpha = \frac{\eta a}{d^{2}}$ and $\Omega = bc^{2}\eta $.
%
To study the stability of the model, we substitute $r(t) = e^{st}$ in \eqref{simpl_d_d_r(t)} and solve for $s$. We obtain a characteristic equation as follows
\begin{align}
\Large
\label{char_eq}
     s^{2} + (2-e^{-sd} + e^{-sd_{s}})\gamma s + \alpha e^{-sd} + \Omega e^{-sd_{s}} = 0
\end{align}
For this dynamical system to be stable, the roots of \eqref{char_eq} should have a negative real term. Note that this equation will have infinite roots, and hence, we seek conditions, i.e., TCP AIMD and scenario parameters, under which the system is always stable.

Classically, we consider AIMD with $a = 1$ and $ b = \frac{1}{2}$, hence, $ \gamma, \alpha$ and $\Omega $ are, $\frac{c}{2 + c^{2}d^{2}}$, $\frac{\eta}{d^{2}} $ and $ \frac{\eta c^{2}}{2}$ respectively. Substituting in \eqref{char_eq} and reorganizing we obtain $\eta$ in \eqref{eta}.
Now, for the dynamical system to be stable, $\eta $ should be greater than 0 under the condition that $s$ has a negative real part.
For simplicity, assume that $s$ is real; hence for $\eta>0$ we obtain the condition $ (e^{-sd}-e^{sd_{s}} - 2)s \gamma - s^{2}$ $< \ 0$. After some algebraic manipulation, and using
a power series expansion and solving for $s$ we obtain

{
\begin{align}
\Large
\label{s_star}
    s^{*} = \frac{ - \sqrt{5d^{2} \gamma^{2} - 2d\gamma^{2}d_{s} + 2d \gamma - 3\gamma^{2}d_{s}^{2} - 2\gamma d_{s} + 1} - d\gamma + \gamma d_{s} - 1} {\gamma (d^{2} - d_{s}^{2})}
\end{align}
}
Substituting any value of s less than $s^{*}$ in \eqref{eta} insures stability.
\end{proof}

\balance
\bibliographystyle{IEEEtran}
\bibliography{ExtFC.bib}


\end{document}